\documentclass[]{interact}

\usepackage{epstopdf}
\usepackage{amsmath,amssymb,amsthm}
\usepackage{graphicx}
\usepackage{algorithmic}
\usepackage{array}
\usepackage{textcomp}
\usepackage{caption}
\usepackage{xcolor}
\usepackage{subcaption}
\usepackage[numbers,sort&compress]{natbib}
\bibpunct[, ]{[}{]}{,}{n}{,}{,}
\renewcommand\bibfont{\fontsize{10}{12}\selectfont}

\theoremstyle{plain}
\newtheorem{theorem}{Theorem}[section]
\newtheorem{lemma}[theorem]{Lemma}

\theoremstyle{definition}
\newtheorem{definition}[theorem]{Definition}
\newtheorem{example}[theorem]{Example}

\theoremstyle{remark}

\newcommand{\R}{\mathbb{R}}

\newcommand{\system}{\mathcal{S}}

\begin{document}


\title{Novel Stability Criteria for Discrete and Hybrid Systems via Ramanujan Inner Products}

\author{
\name{Shyam Kamal\textsuperscript{a}\thanks{Shyam Kamal. Email: shyamkamal.eee@iitbhu.ac.in}, Sunidhi Pandey\textsuperscript{a}\thanks{Sunidhi Pandey. Email: sunidhipandey.rs.eee20@itbhu.ac.in}, Thach Ngoc Dinh\textsuperscript{b}\thanks{Thach Ngoc Dinh. Email: ngoc-thach.dinh@lecnam.net}, and Cao Thanh Tinh\textsuperscript{c,d}\thanks{Cao Thanh Tinh. Email: tinhct@uit.edu.vn}}
\affil{\textsuperscript{a}Department of Electrical Engineering, IIT(BHU) Varanasi, UP, India; \textsuperscript{b}Cedric-Laetitia, Conservatoire National des Arts et Métiers, 292 rue St-Martin, Paris Cedex 03, 75141, France; \textsuperscript{c}Vietnam National University, Ho Chi Minh City, Vietnam;
\textsuperscript{d}Department of Mathematics and Physics, University of Information Technology, Ho Chi Minh City, Vietnam}
}

\maketitle

\begin{abstract}
This paper introduces a Ramanujan inner product and its corresponding norm, establishing a novel framework for the stability analysis of hybrid and discrete-time systems as an alternative to traditional Euclidean metrics. We establish new $\epsilon$-$\delta$ stability conditions that utilize the unique properties of Ramanujan summations and their relationship with number-theoretic concepts. The proposed approach provides enhanced robustness guarantees and reveals fundamental connections between system stability and arithmetic properties of the system dynamics. Theoretical results are rigorously proven, and simulation results on numerical examples are presented to validate the efficacy of the proposed approach.
\end{abstract}


\section{Introduction}
\label{sec:introduction}

Stability analysis is a cornerstone of control theory, with Lyapunov's direct method, input-to-state stability (ISS), and small-gain theorems serving as fundamental tools for characterizing system behavior \cite{khalil2002nonlinear,goebel2009hybrid,Zhai2001}. Traditional stability criteria \cite{Xu2008,Herrmann2006} predominantly employ Euclidean inner products and norms, which measure energy and convergence in a purely geometric sense. While powerful, these metrics may overlook structural properties of signals and dynamics that arise in modern control systems, especially those involving hybrid switching, sampled-data control, or discrete scheduling mechanisms. In such systems, the underlying state evolution may possess arithmetic or combinatorial characteristics, such as periodic switching sequences or residue-class-dependent dynamics, which are not optimally captured by Euclidean norms.

Motivated by this observation, we propose a novel framework for stability analysis based on Ramanujan inner products and associated Ramanujan norms. Ramanujan sums $c_q(n)$, introduced by Srinivasa Ramanujan in 1918 \cite{hardy1918asymptotic}, are arithmetic sums of exponential functions that capture divisibility and residue-class information. These sums have been widely studied in number theory and have found applications in signal processing \cite{planat2002ramanujan,apostol2013introduction}, where they are used to reveal hidden periodicities and quasi-periodic structures in integer-indexed data. The key idea of this paper is to leverage Ramanujan sums to construct a new inner product space for system trajectories, thereby defining a norm that is sensitive to arithmetic structure.

This perspective is particularly relevant for hybrid and discrete-time systems, where the behavior of solutions is strongly influenced by discrete jumps and sampling instants. By employing a Ramanujan-inspired inner product, we obtain stability criteria that are tuned to periodic, almost-periodic, and arithmetic progression-based dynamics. Our results include the Ramanujan stability criterion for discrete-time and hybrid systems, giving small-gain-like conditions expressed in the Ramanujan metric.

The contributions of this paper are threefold: (a) we introduce the Ramanujan inner product and associated norm into the control-theoretic context, (b) we develop Lyapunov-based and ISS-like stability conditions using this norm, and (iii) we demonstrate that systems that are Ramanujan stable enjoy enhanced robustness against structured disturbances that exhibit arithmetic patterns, such as periodic packet drops or cyclostationary noise.

This paper is organized as follows. Section \ref{sec:prelim} presents preliminaries. Section \ref{sec:ramin} establishes notation for Ramanujan inner products and norms. Section \ref{sec:stbdis} presents Ramanujan stability criterion for discrete-time nonlinear systems with structured arithmetic disturbances. Section \ref{sec:hybrid_framework} extends these results to hybrid systems and also establishes Ramanujan-Lyapunov stability conditions for hybrid systems. Section \ref{nex} is dedicated to the validation of the proposed formulation using simulation results. Finally, the conclusion of this work is stated.
\section{Mathematical Preliminaries}
\label{sec:prelim}
\subsection{Classical Stability Notions}
\subsubsection{Stability of Discrete-Time Nonlinear Systems}
Consider the following nonlinear discrete-time system 
\begin{equation} \label{dsys}
x_{k+1} = f(x_k,u_k), \qquad k \in \mathbb{N}
\end{equation}
where $x_k \in \mathcal{X} \subseteq \mathbb{R}^n$ is the state vector, $u_k \in \mathcal{U} \subseteq \mathbb{R}^m$ is the input, and $f:\mathcal{X} \times \mathcal{U} \to \mathcal{X}$ is a nonlinear state transition map assumed to be locally Lipschitz. The equilibrium point be $x^\star = 0$, satisfying $f(0,0) = 0$.
\begin{definition} \cite{haddad2011nonlinear}
   Consider the dynamics \eqref{dsys}. The zero solution is said to be stable if for every $\epsilon > 0$, there exists a $\delta > 0$ such that $\|x_0\|_2 < \delta \quad \Rightarrow \quad \|x_k\|_2 < \epsilon, \qquad \forall k \in \mathbb{N}.$ If, in addition to being stable, the state trajectory converges to the equilibrium as $k \to \infty$, i.e., $\lim_{k \to \infty} x_k = 0,$ then the equilibrium point is said to be asymptotically stable. 
\end{definition}
\subsubsection{Stability of Nonlinear Hybrid Systems}
Consider the following nonlinear hybrid system
\begin{equation} \label{hbdsys}
\system: 
\dot{x} = f(x,u), ~ x \in C, \qquad
x^{+} = g(x,u), ~ x \in D
\end{equation}
where $x \in \mathcal{X} \subseteq \mathbb{R}^n$ denotes the system state and $u \in \mathcal{U} \subseteq \mathbb{R}^m$ represents the input. The set $C \subseteq \mathcal{X}$ is called the flow set and characterizes the states where the system evolves according to the continuous-time dynamics governed by the locally Lipschitz function $f:\mathcal{X} \times \mathcal{U} \to \mathbb{R}^n$. The set $D \subseteq \mathcal{X}$ is the jump set which specifies the states where a discrete transition occurs, with the jump map $g:\mathcal{X} \times \mathcal{U} \to \mathcal{X}$ determining the state immediately after the jump. Solutions of the hybrid system are defined on a hybrid time domain $\mathcal{T} \subset \mathbb{R}_{\geq 0} \times \mathbb{N}$, where each solution trajectory $x(t,j)$ is parameterized by the pair $(t,j)$. Here $t$ represents the usual continuous time during flows, while $j$ counts the number of discrete jumps that have occurred up to time $t$. Here, The classical notion of stability is defined with respect to an equilibrium point $x^\star = 0$, and is expressed using the Euclidean norm.

\begin{definition}\cite{goebel2009hybrid}
Consider the hybrid system \eqref{hbdsys}. The equilibrium point $x^\star = 0$ is said to be stable if for every $\epsilon > 0$, there exists a $\delta > 0$ such that $\|x(0,0)\|_2 < \delta \quad \Rightarrow \quad \|x(t,j)\|_2 < \epsilon, \quad \forall (t,j) \in \mathcal{T}.$ 
\end{definition}

This definition naturally extends the well-known concept of Lyapunov stability from continuous-time systems to hybrid systems by explicitly considering the hybrid time domain. If, in addition, the solution satisfies $\lim_{(t,j)\to\infty} x(t,j) = 0$, the equilibrium is said to be asymptotically stable. 
\section{Ramanujan Inner Product} \label{sec:ramin}
In this section, we define the Ramanujan inner product and associated norm based on Ramanujan's sum.
\begin{definition}[Ramanujan Sum]  For a positive integer $q$ and an integer $n$, the Ramanujan sum is defined as
\begin{equation}
c_q(n) = \sum_{\substack{k=1 \\ (k,q)=1}}^{q} e^{2\pi i \frac{k}{q} n}
\end{equation}
where $(k,q)$ denotes the greatest common divisor (gcd) of $k$ and $q$. In words, the sum is taken over all integers $k$ between $1$ and $q$ that are coprime to $q$. The function $c_q(n)$ is periodic in $n$ with period $q$, and it is always a real-valued number despite being expressed as a sum of complex exponentials.
\end{definition}
\begin{definition}\cite{ramanujan1918certain}[Ramanujan Inner Product] \label{definn}
    Let $a=\{a_n\}$ and $b=\{b_n\}$ be two sequences in $\ell^2(\mathbb{Z})$, i.e., $\sum_{n=-\infty}^{\infty} |a_n|^2 < \infty$ and $\sum_{n=-\infty}^{\infty} |b_n|^2 < \infty$. The Ramanujan inner product between $a$ and $b$ is defined as 
    \begin{align}
        \langle a, b \rangle_{\mathrm{R}}
= \lim_{Q \to \infty} \frac{1}{Q} 
\sum_{q=1}^{Q} \frac{1}{q} 
\sum_{r=0}^{q-1} c_q(r)\sum_{\substack{n \in \mathbb{Z} \\ n \equiv r \, (\mathrm{mod}\, q)}} 
a_n \, \overline{b_n}.
    \end{align}
    where $c_q(r)$ is the classical Ramanujan sum.
\end{definition}
This definition groups the terms of each sequence according to their residue class modulo $q,$ then weights them with the arithmetic kernel $c_q(r)$. In this way, the inner product performs a number-theoretic correlation, analogous to Fourier inner products but aligned with arithmetic periodicities rather than trigonometric frequencies. 

The associated Ramanujan norm of a sequence $a$ is defined in the usual way by
\begin{equation}
\|a\|_{\mathrm{R}} = \sqrt{\langle a, a \rangle_{\mathrm{R}}},
\end{equation}
which provides a measure of the ``energy" of the sequence in the Ramanujan sense. 
\subsection{Properties of the Ramanujan Norm}
\begin{theorem}
Let $\mathcal S$ be the class of complex sequences $a=(a_n)_{n\in\mathbb Z}$ that admit a Ramanujan expansion
$a(n) = \sum_{d\ge 1} \alpha_d\,c_d(n)$,
with coefficient sequence $(\alpha_d)_{d\ge1}$ satisfying the weighted square--summability condition
$\sum_{d\ge1} \varphi(d)\,|\alpha_d|^2 < \infty,$
where $c_d(n)$ denotes the classical Ramanujan sum and $\varphi$ denotes Euler’s totient function.
For $a,b\in\mathcal S$ with coefficients $(\alpha_d)$ and $(\beta_d)$, define the Ramanujan inner product by Parseval-like identity $\langle a,b\rangle_{\mathrm R} := \sum_{d\ge1}\varphi(d)\,\alpha_d\,\overline{\beta_d}.$
The associated function $\|a\|_R:=\sqrt{\langle a,a \rangle_R}$ defines a norm on $\mathcal{S}$. In particular, this norm satisfies (a) Positivity and definiteness: $\|a\|\geq 0$, with equality iff $a\equiv 0$, (b) Homogeneity: $\|\lambda a\|_R=\lambda \|a\|_R$ for all scalars $\lambda$, (c) Triangle inequality: $\|a+b\|_R\leq \|a\|_R+\|b\|_R$.
\end{theorem}
\begin{proof}
    Let $\mathcal{S}$ be the class of complex sequences $a=(a_n)_{n \in \mathbb{Z}}$ admitting Ramanujan expansions $a(n)=\sum_{d\ge1}\alpha_d\,c_d(n)\quad(n\in\mathbb Z),$ with coefficient sequence $(\alpha_d)_{d\ge1}$ satisfying $\sum_{d\ge1}\varphi(d)\,|\alpha_d|^2<\infty,$ where $c_d(n)$ are the classical Ramanujan sums and $\varphi$ is Euler’s totient function. For two sequences $a,b\in\mathcal S$, let their coefficients be $(\alpha_d)$ and $(\beta_d)$, respectively.

    Parseval-Like Formula: Define for finite $Q$, $F_Q(a,b)
:=\frac{1}{Q}\sum_{q=1}^{Q}\frac{1}{q}\sum_{r=0}^{q-1} c_q(r)
\sum_{\substack{n\in\mathbb Z\\ n\equiv r\pmod q}} a_n\overline{b_n}$. For each fixed $Q$ all sums here are finite (only finitely many nonzero terms contribute for any fixed finite data truncation), so substitution of the finite sums is algebraically legitimate. Substituting the Ramanujan expansions of $a$ and $b$, for fixed $Q$ we obtain $
F_Q(a,b)=\sum_{d\ge1}\sum_{e\ge1}\alpha_d\overline{\beta_e}\,L_Q(d,e),$
where $L_Q(d,e):=\frac{1}{Q}\sum_{q=1}^{Q}\frac{1}{q}\sum_{r=0}^{q-1}c_q(r)
\sum_{\substack{n\in\mathbb Z\\n\equiv r\pmod q}} c_d(n)\,c_e(n).$ For fixed $d,e$, $L(d,e):=\lim_{Q\rightarrow \infty}L_Q(d,e)$ exists and equals $\varphi(d)$ if $d=e$, and $0$, if $d \neq e$ (This is the classical orthogonality identity for Ramanujan sums; it follows by writing each $c_l(n)=\sum_{a_{(a,l)=1}}e^{2\pi ian/l},$ expanding and using additive-character orthogonality)(see \cite{nicol1962some}).Hence, \[
\lim_{Q\to\infty}F_Q(a,b)=\sum_{d\ge1}\sum_{e\ge1}\alpha_d\overline{\beta_e}\,L(d,e)
=\sum_{d\ge1}\varphi(d)\,\alpha_d\overline{\beta_d}.
.\] To justify passing the limit inside the double sum, we need a dominated-convergence type bound uniform in $Q$. We obtain such a bound from Cauchy–Schwarz in the weighted $l^2(\varphi)$ space: $|F_Q(a,b)|\leq \sum_{d \geq 1}\sum_{e\geq 1}|\alpha_d||\beta_e||L_Q(d,e)|,~~\text{for each $Q$}.$ There exists a constant $C_{d,e}$ such that $|L_Q(d,e)|\leq C_{d,e}.$ Thus the limit defining the Ramanujan inner product exists for $a,b \in \mathcal{S}$ and equals the coefficient-side Parseval-Like Formula $\langle a,b\rangle_R=\sum_{d\geq 1}\varphi(d)\alpha_d \overline{\beta_d}.$ 

(a) Positivity and definiteness:  $a \in \ell^2(\mathbb{Z})$ and Ramanujan norm is defined as $\|a\|^2=\lim_{Q\rightarrow \infty} \frac{1}{Q}\sum_{q=1}^Q \frac{1}{q}\sum_{r=0}^{q-1}c_q(r)\sum_{\substack{n\in\mathbb Z\\ n\equiv r\pmod q}} |a_n|^2$. Each inner sum $\sum_{n\equiv r\pmod q}|a_n|^2$ is nonnegative. The Ramanujan sum $c_q(r)$ can be positive or negative, but the orthogonality structure of the averaging over $q$ ensures that the overall quadratic form is nonnegative. To see this, fix any finite truncation $Q$. Define $F_Q(a,a):=\frac{1}{Q}\sum_{q=1}^Q \frac{1}{q}\sum_{r=0}^{q-1}c_q(r)\sum_{\substack{n\in\mathbb Z\\ n\equiv r\pmod q}} |a_n|^2$. For each finite $Q$ this can be rewritten as $F_Q(a,a)=\frac{1}{Q}\sum_{q=1}^Q \frac{1}{q} \sum_{n \in \mathbb{Z}}|a_n|^2c_q(n)$. From the classical results on Ramanujan's sum based on the well-defined Cesàro limit (see \cite{planat2002ramanujan}) $\frac{1}{Q}\sum_{q=1}^{Q} \frac{1}{q}c_q(n)\rightarrow w(n)\geq 0$. If $\|a\|_R=0,$ then every weighted term $w(n)|a_n|^2=0.$ Since each $w(n)\geq 0$ and at least one weight is strictly positive for each congruence class, it follows that $|a_n|^2=0$ for all $n$. Hence, $a_n=0$ for all $n$. Conversely, if $a_n=0$ for all $n$, clearly $\|a\|_R^2=0.$

(b) Homogeneity: If $\lambda \in \mathbb{C},$ then coefficients of $\lambda a$ are $\lambda \alpha_d$ and so  $\|\lambda a\|_R^2=\sum_d \varphi(d) |\lambda \alpha_d|^2=\lambda^2\|a\|_R^2.$

(c) Triangle inequality: $\|a+b\|_R^2=\langle a+b, a+b\rangle_R=\langle a,a\rangle_R+\langle a,b\rangle_R+\langle b,a\rangle_R+\langle b,b\rangle_R=\|a\|_R^2+2 \mathfrak{R}\langle a,b\rangle_R+\|b\|_R^2$, where $\mathfrak{R}z$ denotes the real part of $z$. As $2 \mathfrak{R}\langle a,b\rangle_R \leq 2|\langle a,b \rangle_R|$ and $2|\langle a,b \rangle_R|\leq \|a\|_R \|b\|_R$. Then we can write $\|a+b\|_R^2\leq \|a\|_R^2+2 \|a\|_R \|b\|_R+\|b\|_R^2=(\|a\|_R+\|b\|_R)^2.$ Hence, $\|a+b\|_R\leq \|a\|_R+\|b\|_R.$
\end{proof}
\subsection{Comparison with the Euclidean Norm}
\begin{lemma}
    Let $x=(x_n)_{n \in \mathbb{Z}} \in l^2 (\mathbb{Z})$, and $\langle x,x \rangle_R=\lim_{Q\rightarrow \infty} \frac{1}{Q}\sum_{q=1}^Q \frac{1}{q}\sum_{r=0}^{q-1}c_q(r)\sum_{n\equiv r\pmod q} |x_n|^2$ exists then $\|x\|_R\leq \|x\|_2,$ 
\end{lemma}
\begin{proof}
   For fixed $q$, residue classes partition $\mathbb{Z}$, so $\sum_r \sum_{n \equiv r} |x_n|^2=\|x\|^2$. Since, $|c_q(r)|\leq \varphi(q)$ and $\varphi(q)/q \leq 1$, \[\frac{1}{q}\sum_r c_q(r) \sum_{n\equiv r}|x_n|^2 \leq \|x\|^2.\]
    Averaging and taking limits gives $\|x\|_R^2\leq \|x\|_2^2$.\\
    \end{proof}
\subsection{Ramanujan Gain Function}
\begin{definition}[Ramanujan Gain Function]
Let $\|\cdot\|_R$ denote the Ramanujan norm induced by the Ramanujan inner product on the space of admissible signals. 
The \emph{Ramanujan gain function} is the map
$\gamma_R : \mathbb{R}_{\ge 0} \to \mathbb{R}_{\ge 0},$
\[\gamma_R(r) := \sup\Bigl\{\,\|x(t)\|_R \;\big|\; 
\|x(0)\|_R \le r,\; t \ge 0 \Bigr\},
\]
where $x(\cdot)$ denotes a trajectory of a closed-loop or open-loop system under consideration, and the supremum is taken over all admissible trajectories starting in the Ramanujan ball of radius $r$.
\end{definition}
\section{Main Stability Theorem}
\label{sec:stbdis}
\subsection{Ramanujan Stability Criterion}
\begin{theorem}
Consider a discrete-time nonlinear system
$x_{k+1} = f(x_k), \qquad k \in \mathbb{N}, \quad x_k \in \mathbb{R}^n.$
Fix $q \in \mathbb{N}$ and let $\|\cdot\|$ be Ramanujan norm. Suppose there exists a class $\mathcal{K}$ function $\alpha: \mathbb{R}_{
\geq 0
} \rightarrow \mathbb{R}_{\geq 0}$, constant $M \geq 0$ and $r \in (0,1)$ such that 
\begin{align}
    \|x_k\|_R \leq \alpha(\|x_0\|_R)+M\sum_{j=0}^{k-1}\frac{c_q(k-j)}{q}r^{k-j}\|x_j\|_R,~~\forall~k\geq 0,
\end{align}
Define $G:=M \sum_{n=0}^{\infty} \frac{c_q(n)}{q}r^n$. If $G<1$ then $\sup_{k\geq 0} \|x_k\|_R \leq \frac{\alpha (\|x_0\|_R)}{1-G}$, and in particular the origin is stable in the $\|\cdot\|_R$ sense.  
\end{theorem}
\begin{proof}
    Uniform Bound: For $k\geq 0$, define $S_k:=\max_{0\leq j \leq k}\|x_j\|_R$ and $g(n):=M\frac{|c_q(n)|}{q}r^n$, then \begin{align}
        \|x_k\|_R \leq \alpha(\|x_0\|_R)+\sum_{j=0}^{k-1}g(k-j)\|x_j\|_R,
    \end{align}
    Since $\|x_j\|_R\leq S_{k-1}$ for every $j \in \{0,\ldots,k-1\}$ then we can write $ \|x_k\|_R \leq \alpha(\|x_0\|_R)+GS_{k-1}.$ Hence, for every $k \geq 1$, $S_k=\max\{S_{k-1}, \|x_k\|_R\}.$ As $G<1$ and $\alpha(\|x_0\|_R)\geq 0$, it follows that $S_k \leq \alpha(\|x_0\|_R)+GS_{k-1}.$ Let $S:=\sup_{k\geq 0}S_k,$ then $S\leq \alpha(\|x_0\|_R)+GS \implies S\leq \frac{\alpha (\|x_0\|_R)}{1-G}.$ Therefore, $\sup_{k\geq 0}\|x_k\|_R=S < \infty,$ therefore explicit uniform bound exists.

    Ramanujan Stability: Fix an arbitrary $\epsilon >0$. Let $\delta>0$ be such that $\frac{\alpha(\delta)}{1-G}< \epsilon$. If $\|x_0\|_R< \delta$ then $\sup_{k \geq 0} \|x_k\|_R\leq \frac{\alpha (\|x_0\|_R)}{1-G}\leq \frac{\alpha(\delta)}{1-G}<\epsilon$. This is the $\epsilon-\delta$ definition of stability of the origin in the $\|\cdot\|_R$ sense.
    \end{proof}

\subsection{Robustness under Arithmetic Disturbances} \label{rrad}

\begin{theorem}
Consider the discrete-time perturbed nonlinear system
 $x_{k+1} = f(x_k) + w_k, \qquad k \in \mathbb{N},$
where $w_k \in \mathbb{R}^n$ is an exogenous disturbance sequence. 
Suppose the nominal system $x_{k+1} = f(x_k)$ is Ramanujan stable. 
Then, for any disturbance sequence $\{w_k\}$ whose Ramanujan norm is uniformly bounded, i.e.,  $W := \sup_{k \ge 0} \|w_k\|_R < \infty,$
the state trajectories $\{x_k\}$ remain uniformly bounded in the Ramanujan norm. 
\end{theorem}

\begin{proof}
By the Ramanujan stability criterion, there exist a class-$\mathcal{K}$ function $\alpha$, constants $M>0$, $r\in(0,1)$, and a fixed $q\in\mathbb{N}$ such that the nominal trajectories of the system satisfy
\begin{equation}\label{eq:nominal_stability}
    \|x_k\|_R \le \alpha(\|x_0\|_R)
    + M \sum_{j=0}^{k-1} \frac{|c_q(k-j)|}{q} r^{k-j} \|x_j\|_R,
    \quad \forall k\ge0.
\end{equation}
For the perturbed system $x_{k+1}=f(x_k)+w_k$, the dynamics can be viewed as the nominal part plus an additive disturbance. 
By applying the same inequality recursively and noting the linearity of addition in the argument of the Ramanujan norm, we obtain
\begin{align}\label{eq:perturbed_stability}
\nonumber    \|x_k\|_R 
    \le& \alpha(\|x_0\|_R)
    + M \sum_{j=0}^{k-1} \frac{|c_q(k-j)|}{q} r^{k-j} \|x_j\|_R\\&~~~~~~~~
    + M \sum_{j=0}^{k-1} \frac{|c_q(k-j)|}{q} r^{k-j} \|w_j\|_R.
\end{align}
Define $S_k := \max_{0 \le j \le k} \|x_j\|_R$ and let
$G := M \sum_{n=0}^\infty \frac{|c_q(n)|}{q} r^n.$
Since $r\in(0,1)$ and $|c_q(n)|\le q$, the series defining $G$ converges absolutely, and $G<\infty$. 
Then, from \eqref{eq:perturbed_stability},
$\|x_k\|_R \le \alpha(\|x_0\|_R) + G S_{k-1} + G W,$
and thus
$S_k = \max\{S_{k-1}, \|x_k\|_R\} \le \alpha(\|x_0\|_R) + G S_{k-1} + G W.
$ Taking the supremum over $k$ and noting that $S := \sup_{k\ge0}S_k$ satisfies $S \le \alpha(\|x_0\|_R) + G S + G W$, we obtain
\[
S \le \frac{\alpha(\|x_0\|_R) + G W}{1 - G}.
\]
Hence, $\sup_{k\ge0} \|x_k\|_R = S < \infty$, proving uniform boundedness of the trajectories in the Ramanujan norm. 

Moreover, if the nominal system is globally asymptotically Ramanujan-stable (i.e., $\gamma_R(\|x_0\|_R) \to 0$ as $\|x_0\|_R \to 0$), then the perturbed trajectory converges to a bounded invariant neighborhood of the origin whose radius is proportional to $G W / (1-G)$. 
\end{proof}
\paragraph{Control-Theoretic Interpretation.}
Ramanujan stability guarantees a form of disturbance attenuation tailored to signals with arithmetic structure. 
The weighting by $c_q(\cdot)$ in the gain kernel allows the system to selectively attenuate disturbances exhibiting residue-class or modular periodicity. 
Hence, systems stable in the Ramanujan metric maintain bounded trajectories even under persistent periodic, quasi-periodic, or cyclostationary disturbances—phenomena not well captured by Euclidean or $\ell^2$ metrics. 
This corresponds to a number-theoretic small-gain property, ensuring bounded-input–bounded-state robustness with respect to disturbances whose energy is bounded in the Ramanujan sense.
\section{Hybrid System Framework and Stability Definitions}
\label{sec:hybrid_framework}

\subsection{Hybrid System Model}

We consider a general hybrid system $\mathcal{H}$ in the formalism of~\cite{goebel2009hybrid}, described by
\begin{align}
\mathcal{H}: \quad & \dot{x} = f(x), && x \in C, \label{eq:flow_dynamics}\\
& x^+ = g(x), && x \in D, \label{eq:jump_dynamics}
\end{align}
where $x \in \mathcal{X} \subseteq \mathbb{R}^n$ is the system state. The \emph{flow set} $C \subseteq \mathcal{X}$ specifies the region where continuous evolution occurs according to the flow map $f: C \to \mathbb{R}^n$, and the \emph{jump set} $D \subseteq \mathcal{X}$ specifies the region where discrete transitions occur according to the jump map $g: D \to \mathcal{X}$. A hybrid time domain $\operatorname{dom} x \subseteq \mathbb{R}_{\ge 0} \times \mathbb{N}$ is a union of intervals of the form
$\operatorname{dom} x = \bigcup_{j=0}^{J-1} \big([t_j, t_{j+1}], j\big),$ where $0 = t_0 \le t_1 \le t_2 \le \dots$ and $J$ may be finite or infinite. A hybrid arc $x: \operatorname{dom} x \to \mathcal{X}$ is a function such that, for each $j$, the mapping $t \mapsto x(t,j)$ is locally absolutely continuous on $[t_j, t_{j+1}]$ and satisfies the hybrid dynamics~\eqref{eq:flow_dynamics}–\eqref{eq:jump_dynamics} almost everywhere.
\subsection{Ramanujan Stability for Hybrid Systems}

A key challenge in hybrid stability analysis lies in unifying the effects of continuous flows and discrete jumps.  
While the Euclidean norm captures geometric magnitudes, it does not encode the \emph{temporal or arithmetic structure} of jump occurrences.  
The Ramanujan norm, defined via the Ramanujan inner product, provides an alternative metric that is sensitive to such arithmetic periodicities, enabling refined stability characterizations for systems exhibiting event-triggered or number-theoretic timing patterns. We now extend classical Lyapunov-based hybrid stability definitions to the Ramanujan metric.

\begin{definition}
Let $\mathcal{A} \subset \mathcal{X}$ be compact and let $\|\cdot\|_R$ denote the Ramanujan norm induced by the Ramanujan inner product. The set $\mathcal{A}$ is said to be Ramanujan stable for the hybrid system $\mathcal{H}$ if for every $\epsilon > 0$ there exists $\delta > 0$ such that every solution $x$ to $\mathcal{H}$ with $\|x(0,0)\|_R < \delta$ satisfies
 $\|x(t,j)\|_R < \epsilon, 
\quad \forall (t,j) \in \operatorname{dom} x.$

If, in addition, there exists $\delta_0 > 0$ such that for all initial conditions with $\|x(0,0)\|_R < \delta_0$, $\lim_{(t+j)\to\infty,\, (t,j)\in\operatorname{dom}x} \operatorname{dist}_R(x(t,j), \mathcal{A}) = 0.$
Then $\mathcal{A}$ is said to be asymptotically Ramanujan stable.
\end{definition}

Here, $\operatorname{dist}_R(x,\mathcal{A}) := \inf_{y \in \mathcal{A}} \|x - y\|_R$ denotes the distance in the Ramanujan metric.  
\subsection{Ramanujan Lyapunov Functions for Hybrid Systems}

To analyze Ramanujan stability, we employ a Lyapunov-like function $V : \mathcal{X} \to \mathbb{R}_{\ge 0}$ satisfying:
\begin{align*}
\alpha_1(\|x\|_R) \le V(x) \le \alpha_2(\|x\|_R), & \quad \forall x \in \mathcal{X}, \\
\dot{V}(x) \le -\alpha_3(\|x\|_R), & \quad \forall x \in C,\\
V(g(x)) - V(x) \le -\alpha_4(\|x\|_R), & \quad \forall x \in D,
\end{align*}
for class-$\mathcal{K}_\infty$ functions $\alpha_i$, $i=1,\dots,4$.  
Under these conditions, the set $\mathcal{A} = \{x : V(x)=0\}$ is asymptotically Ramanujan stable.  
This follows by extending standard hybrid Lyapunov arguments~\cite{khalil2002nonlinear} to the Ramanujan norm, leveraging the norm’s equivalence to the Euclidean norm on compact sets.
\subsection{Ramanujan-Lyapunov Stability for Hybrid Systems}
\begin{theorem}
    Consider the hybrid system $\mathcal{H}$. Suppose there exists a continuously differentiable function $V:\mathbb{R}^n \to \mathbb{R}_{\ge 0}$ (called a \emph{Ramanujan--Lyapunov function}) and class-$\mathcal{K}_\infty$ functions $\alpha_1,\alpha_2,\alpha_3,\alpha_4$ such that:
\begin{enumerate}
\item Positive definiteness: $\alpha_1(\operatorname{dist}_R(x,\mathcal{A})) 
\le V(x) \le 
\alpha_2(\operatorname{dist}_R(x,\mathcal{A})),
\qquad \forall x \in \mathbb{R}^n.$
\item Decrease along flows: $\langle \nabla V(x), f(x,u) \rangle 
\le -\alpha_3(\operatorname{dist}_R(x,\mathcal{A})),
\qquad \forall x \in C.$
\item Decrease across jumps: $V(g(x,u)) - V(x) 
\le -\alpha_4(\operatorname{dist}_R(x,\mathcal{A})),
\qquad \forall x \in D.$
\end{enumerate}
Then the set $\mathcal{A}$ is asymptotically Ramanujan stable, i.e., $\lim_{(t+j)\to\infty,\, (t,j)\in\operatorname{dom}x} 
\operatorname{dist}_R(x(t,j),\mathcal{A}) = 0.$

\end{theorem}
\begin{proof}
Let $x(t,j)$ be any complete solution of $\mathcal{H}$, defined on a hybrid time domain with unbounded $t+j$. As,
$\frac{d}{dt}V(x(t,j)) \le -\alpha_3(\operatorname{dist}_R(x(t,j),\mathcal{A})).$
Integrating over $[t_0,t]$ yields
\[
V(x(t,j)) \le V(x(t_0,j))
- \int_{t_0}^{t}\alpha_3(\operatorname{dist}_R(x(s,j),\mathcal{A}))\,ds.
\]
Hence $V(x(t,j))$ is nonincreasing during flows and strictly decreasing whenever $x(t,j)\notin \mathcal{A}$.

At a jump instant $(t,j)\in\operatorname{dom}x$ with $x(t,j)\in D$, $V(x(t,j+1)) - V(x(t,j))
\le -\alpha_4(\operatorname{dist}_R(x(t,j),\mathcal{A})),$
implying a nonnegative decrement across jumps.
\paragraph{Lyapunov (Ramanujan) Stability.}
Given any $\epsilon>0$, select $\delta>0$ such that $\alpha_2(\delta)<\alpha_1(\epsilon)$.  
If $\|x(0,0)\|_R<\delta$, then $V(x(0,0))<\alpha_1(\epsilon)$. Since $V$ is nonincreasing, $V(x(t,j))<\alpha_1(\epsilon)$ for all $(t,j)$, hence $\|x(t,j)\|_R<\epsilon$. This establishes Lyapunov (Ramanujan) stability.


\paragraph{Asymptotic Convergence.}
 By compactness of $\mathcal{A}$ and continuity of $\alpha_3,\alpha_4$, there exists $\varepsilon>0$ such that $\alpha_3(\operatorname{dist}_R(x,\mathcal{A}))+\alpha_4(\operatorname{dist}_R(x,\mathcal{A})) > 0$ whenever $\operatorname{dist}_R(x,\mathcal{A})>\varepsilon$.  
But since $V$ strictly decreases whenever $x\notin\mathcal{A}$, this contradicts the convergence of $V(x(t,j))$ to a positive limit.  
Therefore, $V(x(t,j))\to 0 \quad\Rightarrow\quad 
\operatorname{dist}_R(x(t,j),\mathcal{A})\to 0.$
The function $V$ therefore satisfies the hybrid Lyapunov conditions in the Ramanujan metric, ensuring that $\mathcal{A}$ is asymptotically Ramanujan stable.
\end{proof}
\section{Numerical Examples} \label{nex}
This section presents an example of a discrete-time system with an arithmetic disturbance. The example is designed to highlight a key scenario where traditional Euclidean-based stability analysis fails to certify stability, while the proposed Ramanujan-based criterion succeeds, demonstrating its reduced conservatism.

\begin{example}Consider the following linear system:
\begin{equation} \label{dts}
x_{k+1} = A x_k + w_k, \quad A = \begin{bmatrix}
0.8 & 0.1 \\ 0 & 0.8
\end{bmatrix},
\end{equation}
where $w_k \in \R^2$ is an exogenous disturbance signal. The matrix $A$ is Schur stable ($\rho(A) = 0.8 < 1$), so the nominal system ($w_k \equiv 0$) is globally asymptotically stable.

The disturbance is defined to be persistent but arithmetically sparse. Specifically, it is non-zero only when the time index $k$ is a prime number:
\begin{equation}
w_k = 
\begin{cases}
\mathbf{0} & \text{if } k \text{ is not prime}, \\
(0, \, 0.5)^\top & \text{if } k \text{ is prime}.
\end{cases}
\end{equation}
This disturbance is not absolutely summable ($\sum_{k=0}^\infty \|w_k\|_2 = \infty$) due to the density of primes, making classical $\ell_2$-gain or input-to-state stability (ISS) analysis challenging.
\end{example}
\subsubsection{Comparison Between Ramanujan Stability Analysis and Euclidean Analysis}
In the classical Euclidean framework, stability is analyzed using a quadratic ISS–Lyapunov function $V(x)=x^\top Px$, where $P>0$ satisfies $A^\top PA-P=-Q$ for some $Q>0$. The ISS condition $V(x_{k+1})-V(x_k)\leq -\alpha V(x_k)+\sigma \|w_k\|_2^2$, where $\alpha, \sigma>0$. The resulting bound $\|x_k\|_2\leq c \beta^k \|x_0\|_2+\gamma \sum_{j=0}^{k-1}\beta^{k-j-1}\|w_j\|_2$, where $c>0, \gamma>0,$ $ \beta =\sqrt{1-\alpha}<1,$ ensures boundedness but not convergence when the disturbance sequence (here, prime-indexed $w_k$) is non-summable. In contrast, under the Ramanujan framework, stability is governed by $\|x_k\|_R 
    \le \alpha(\|x_0\|_R)
    + M \sum_{j=0}^{k-1} \frac{|c_q(k-j)|}{q} r^{k-j} \|x_j\|_R
    + M \sum_{j=0}^{k-1} \frac{|c_q(k-j)|}{q} r^{k-j} \|w_j\|_R,$ as given in subsection \ref{rrad}. Here, $c_q(k-j)$ down-weights contributions from indices arithmetically uncorrelated with the modulus $q,$ the term involving $\|w_j\|_R$ grows sublinearly even for persistent prime-indexed disturbances. Consequently, $\|x_k\|_R$ converges to zero while the Euclidean norm merely remains bounded, demonstrating that the Ramanujan criterion captures arithmetic sparsity and provides a strictly less conservative stability assessment.
   
\subsubsection{Simulation Results}
    A simulation is performed with the initial condition $x_0=(1,1)^\top,$ and the the trajectories of the Euclidean norm $\|x_k\|_2$ and the Ramanujan norm $\|x_k\|_R$ (for modulus $q=5$) are shown in Fig.~\ref{fig:trajectories}. The results validate the theoretical analysis that the Euclidean norm exhibits recurrent spikes at prime indices and does not converge to zero, whereas the Ramanujan norm shows a smooth monotonic decay toward zero. This demonstrates that the Ramanujan-based analysis effectively captures the arithmetic sparsity of the disturbance, yielding a less conservative and more accurate assessment of asymptotic stability compared to the traditional Euclidean approach.
     \begin{figure}[htbp]
\centering
\includegraphics[height=1.85in,width=3.5in]{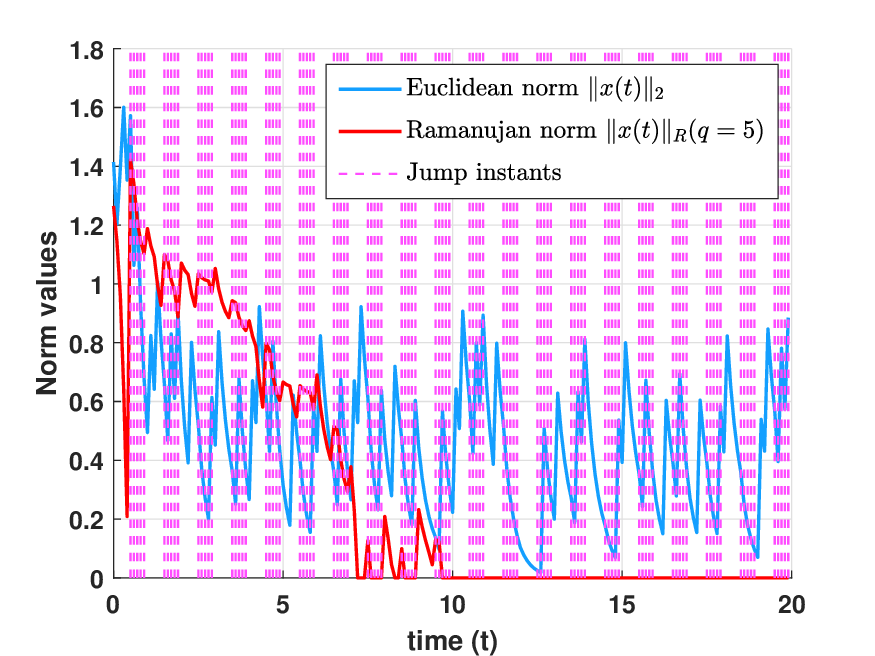}
\caption{Norm of state of system \eqref{dts} showing that the Euclidean norm has persistent spikes due to prime-time disturbances, while the Ramanujan norm decays smoothly to zero.}
\label{fig:trajectories}
\end{figure}
\subsection{Hybrid System Example}
\begin{example}
    Consider a hybrid system $\system$ with state $x = (x_1, x_2) \in \R^2$, representing a mass with position $x_1$ and velocity $x_2$. The system is designed to flow continuously but experience a dissipative jump whenever the continuous time $t$ hits a prime number. The flow set, jump set, and maps are defined as: $C = \{ x \in \R^2 \} \times \{ t \in \R_{\geq 0} : t \notin \mathbb{P} \},
D = \{ x \in \R^2 \} \times \{ t \in \R_{\geq 0} : t \in \mathbb{P} \},$ and
\begin{align} \label{hbd}
 f(x, t) = (x_2, -x_1 - 0.1x_2)^\top, ~~~
g(x, t) = \begin{pmatrix} x_1 \\ 0.6 x_2 \end{pmatrix}. 
\end{align}
Here, $\mathbb{P}$ is the set of prime numbers. The flow dynamics describe a lightly damped oscillator. The jump map imposes an instantaneous dissipation of kinetic energy, but these jumps occur only at prime times. 
\end{example}
\subsubsection{Comparison Between Ramanujan Stability Analysis and Euclidean Analysis} Under the Euclidean framework, $V(x)=\frac{1}{2}(x_1^2+x_2^2)$ represents the oscillator’s total energy, which decreases during both flows and jumps. Along flows $(t \notin \mathbb{P}), \dot{V}(x)=x_1x_2+x_2(-x_1-0.1x_2)=-0.1x_2^2\le 0$, showing non-increasing energy. Across jumps $(t \in \mathbb{P}), V(g(x))-V(x)=\frac{1}{2}x_1^2+\frac{1}{2}(0.6x_2)^2-\frac{1}{2}(x_1+x_2)^2=-0.14x_2^2\le 0,$ confirming a strict energy drop. Thus, $V(x)$ is non-increasing, and the system is Lyapunov stable. However, because prime-triggered jumps become increasingly sparse, the gaps between primes grow logarithmically and the weak continuous damping $-0.1x_2^2$, which makes Euclidean analysis inconclusive on asymptotic convergence and overly conservative. In contrast, the Ramanujan framework captures the arithmetic sparsity of the jump sequence by evaluating stability through the Ramanujan inner product, which correlates the state energy with arithmetic progressions. Sampling the state at prime instants $t_j \in \mathbb{P},$ and defining $V(x)=\frac{1}{2}(x_1^2+x_2^2)$ as a Ramanujan-Lyapunov function, one obtains inequalities of the form $\langle \nabla V(x),f(x)\rangle_R \leq 0$ for flows and $V(g(x))-V(x) \leq -\rho V(x)$ for jumps, with the arithmetic weighting of the Ramanujan sum $c_q(p_j-p_i)$ ensuring that the averaged dissipation over prime-indexed progressions remains sufficient to drive $V(x)\rightarrow 0.$ Consequently, the sequence $\|x^{[j]}\|_R$ converges to zero, establishing asymptotic Ramanujan stability and demonstrating that the proposed criterion provides a less conservative and more accurate characterization of stability than Euclidean analysis.
\subsubsection{Simulation Results}
\begin{figure}[htbp]
\centering
\includegraphics[height=1.85in,width=3.5in]{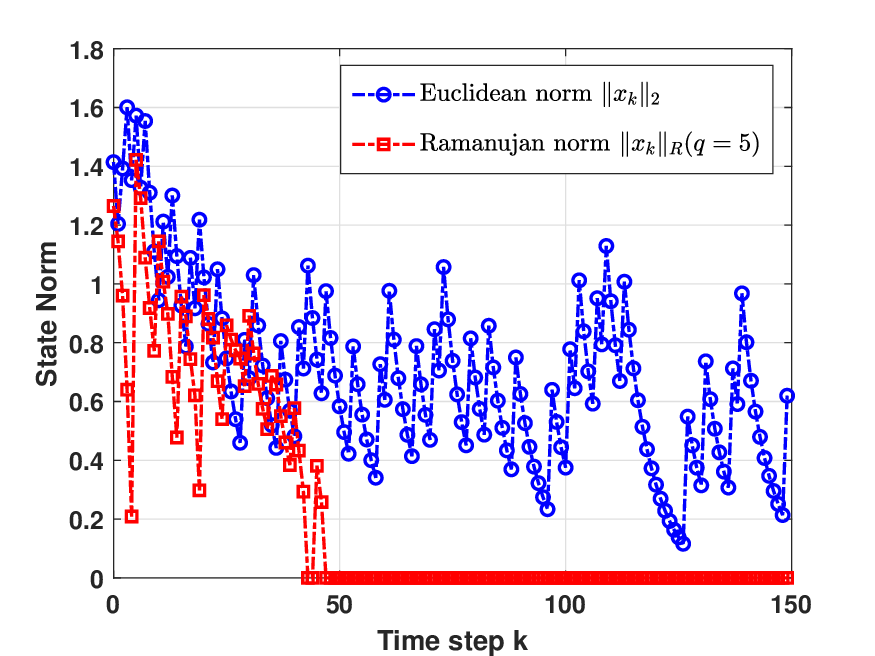}
\caption{Norm of state trajectories of system \eqref{hbd} showing evolution of Euclidean and Ramanujan norm.}
\label{fig:trajectories1}
\end{figure}Simulation is performed for the hybrid system \eqref{hbd}, with initial condition $x_0=(1,1)^\top$, and modulus $q=5.$ Fig.~\ref{fig:trajectories1} shows the evolution of the Euclidean norm $\|x(t)\|_2$ and $\|x(t)\|_R$, and vertical dashed lines indicate jump instants. The Euclidean norm displays noticeable fluctuations and persistent spikes near these jump instants, illustrating slow energy dissipation between sparse prime-triggered events. In contrast, the Ramanujan norm exhibits a smooth monotonic decay over time, confirming asymptotic convergence and demonstrating that the Ramanujan-based analysis effectively captures arithmetic sparsity in jump occurrences, providing a more accurate and less conservative stability characterization than the Euclidean approach.
\subsection{Discussion} Consider a discrete-time stable system $x_{k+1}=Ax_{k}+w_k$, where $w_k=d,$ if $k\equiv r~(\text{mod~m}),$ and $w_k=0,$ otherwise. In the Euclidean framework, the system behaves as an LTI system driven by a periodic input, leading to a steady-state oscillatory response where $\limsup_{k\rightarrow \infty}{\|x_k\|_2}> 0,$ thus classical ISS or $\ell_\infty$ gain analysis concludes boundedness but not asymptotic stability. In contrast, the Ramanujan framework exploits the arithmetic structure of the disturbance. When the chosen Ramanujan modulus $q$ is coprime to $m,$ the Ramanujan sum $c_q(r+mj)$ remains nonzero, and the Ramanujan norm captures the full disturbance energy, reproducing the boundedness behavior consistent with Euclidean analysis. However, when $q$ shares factors with $m,$ particularly for $q=m,$ then $c_m(r+mj)=\mu(m),$ where $\mu(m)$ is the Mobius function. If $\mu(m)=0,$ the disturbance becomes orthogonal to the chosen Ramanujan basis, yielding $\|w_k\|_R=0$ and ensuring asymptotic stability in the Ramanujan norm. Hence, by tuning $q$ one can selectively filter out structured arithmetic disturbances, which demonstrates that the Ramanujan approach provides a refined, non-conservative, and structure-aware stability certification far beyond the capability of traditional Euclidean analysis.

\section{Conclusion}
\label{sec:conclusion}

This paper has introduced a novel stability framework using Ramanujan inner products that provides enhanced stability criteria for hybrid and discrete-time systems. The proposed approach leverages number-theoretic concepts to capture stability properties that are not apparent through traditional Euclidean methods. The main theorems establish rigorous $\epsilon$-$\delta$ stability conditions with improved robustness properties, particularly for systems with arithmetic characteristics in their dynamics.

Future work will focus on computational methods for evaluating Ramanujan norms, applications to networked control systems, and extensions to stochastic hybrid systems.

\section*{Disclosure statement}
No potential conflict of interest was reported by the author(s).

\section*{Data Availability Statement}
Data sharing is not applicable to this article as no new data were created or analyzed in this study.
\bibliographystyle{tfq}
\bibliography{interactapasample}
\section*{Biographical note}
\textbf{Shyam Kamal} received his Bachelor’s degree in Electronics and Communication Engineering from the
 Gurukula Kangri Vishwavidyalaya Haridwar, Uttrakhand, India in 2009, and Ph.D. in Systems and Control Engineering from the IIT Bombay, India in 2014.  He was with the Department of Systems Design and Informatics, Kyushu Institute of Technology, Japan as a Project Assistant Professor, from 2014 to 2016. Currently, he is an associate professor at the Department of Electrical Engineering, IIT (BHU), Varanasi, India. His research spans discrete and continuous higher-order sliding mode control, fractional-order systems, and contraction analysis.
 
\textbf{Sunidhi Pandey} completed her B.Tech degree in Electrical Engineering from KNIT Sultanpur, UP, India, and her M. Tech degree in Power System Engineering from MNNIT Allahabad, India in 2018. She completed her Ph.D. degree in Control System Engineering from the IIT (BHU), Varanasi, India, in 2025. Her research area is nonlinear control, passivity-based control, continuous-time optimization techniques, and sliding mode control.

\textbf{Thach Ngoc
Dinh} received the M.Sc.Res. degree in Automatic Control (with honors) and the Diplôme d’Ingénieur (Master’s degree) in electrical engineering from INSA Lyon, France, in 2011, and the Ph.D. degree in physics from Paris-Saclay University, France, in 2014. His doctoral work was funded by INRIA and conducted at the CAOR Laboratory of Mines Paris–PSL and the L2S Laboratory UMR 8506 of CentraleSupelec. From 2015 to 2016, he was a JSPS Postdoctoral Fellow at the Kyushu Institute of Technology, Japan. He subsequently held a temporary Assistant Professor position at the Polytechnic University of Hauts-de-France, France, from 2016 to 2017. Since September 2017, he has been a tenured Associate Professor at the Conservatoire National des Arts et Métiers (CNAM), France. He was appointed as an Adjunct Professor at Vietnam National University, Ho Chi Minh City (VNUHCM), Vietnam, for the period 2025–2027 under the excellence program VNU350. He has also held visiting positions at the Kyushu Institute of Technology (Japan, February–March, 2019), the Hanse-Wissenschaftskolleg—Institute for Advanced Study and the Carl von Ossietzky University of Oldenburg (Germany, February–May, 2025), Harbin Institute of Technology (China, June–July, 2025), and Vietnam Institute for Advanced Study in Mathematics (Vietnam). His research interests include observer design and robust control. He has served as an Associate Editor for the IEEE Control Systems Letters and the Journal of the Franklin Institute, as an Editor for Nonlinear Dynamics, and as an Editorial Board Member or Guest Editor for several other scholarly journals. He is also a member of the IFAC Technical Committee on Linear Control Systems. Dr. Dinh received the JSPS Postdoctoral Fellowship for North American and European Researchers in 2015 and the Hanse-Wissenschaftskolleg Regular Fellowship in 2023.

\textbf{Cao Thanh Tinh} received the BSc degree in Mathematics from Can Tho University in 2003, the MSc degree in Mathematics from Hanoi National University of Education in 2008, and the PhD degree in Mathematics from Vietnam National University - HCMC, University of Science, Vietnam in 2017. He is currently a senior lecturer in Vietnam National University - HCMC, University of Information Technology. His research interests are in differential equations, robust control, neutral networks, singular systems.


\end{document}